\documentclass[11pt]{article}                                  
\usepackage{booktabs}
\usepackage{algorithm}
\usepackage{algorithmic}

\usepackage{amsthm}
\usepackage{amssymb}
\usepackage{amsmath}
\usepackage[colorlinks=true, urlcolor=blue, citecolor=blue]{hyperref}
\usepackage{multicol}
\usepackage[numbers]{natbib}
\usepackage{xfrac}

\usepackage{qcircuit}
\usepackage{braket}

\newtheorem{definition}{Definition}
\newtheorem{lemma}{Lemma}
\newtheorem{thm}{Theorem}
\newtheorem{cor}{Corollary}

\newcommand{\tr}{\mathsf{tr}}
\newcommand{\Tr}{\mathsf{Tr}}

\newcommand{\corrref}[1]{\hyperref[#1]{Corollary~\ref{#1}}}
\newcommand{\defref}[1]{\hyperref[#1]{Definition~\ref{#1}}}
\newcommand{\secref}[1]{\hyperref[#1]{Section~\ref{#1}}}
\newcommand{\chapref}[1]{\hyperref[#1]{Chapter~\ref{#1}}}
\newcommand{\appref}[1]{\hyperref[#1]{Section~\ref{#1}}}
\newcommand{\thmref}[1]{\hyperref[#1]{Theorem~\ref{#1}}}
\newcommand{\lemref}[1]{\hyperref[#1]{Lemma~\ref{#1}}}
\newcommand{\figref}[1]{\hyperref[#1]{Fig.~\ref{#1}}}
\renewcommand{\eqref}[1]{\hyperref[#1]{Eq. (\ref{#1})}}
\newcommand{\tableref}[1]{\hyperref[#1]{Table (\ref{#1})}}
\newcommand{\algref}[1]{\hyperref[#1]{Algorithm (\ref{#1})}}
\renewcommand{\epsilon}{\varepsilon}

\newtheorem{innercustomgeneric}{\customgenericname}
\providecommand{\customgenericname}{}
\newcommand{\newcustomtheorem}[2]{%
  \newenvironment{#1}[1]
  {%
   \renewcommand\customgenericname{#2}%
   \renewcommand\theinnercustomgeneric{##1}%
   \innercustomgeneric
  }
  {\endinnercustomgeneric}
}

\newcustomtheorem{customthm}{Theorem}
\newcustomtheorem{customlemma}{Lemma}

\begin{document}                                               

\begin{center}                                                 

{\bf                                                           
%
Differential Privacy Amplification in Quantum and Quantum-inspired Algorithms
\\ 
}                                                             
\vspace{2ex}                                                  
Armando Angrisani$^{1}$, 
Mina Doosti$^{2}$,
Elham Kashefi$^{1,2}$
\\
\vspace{2ex}                                                  
{\sl                                                          
%
$^{1}$LIP6, CNRS, Sorbonne Université, 75005 Paris, France\\
$^{2}$School of Informatics, University of Edinburgh, EH8 9AB Edinburgh, United Kingdom
}                                                             
{\small \sl                                                   
%

}                                                             
\end{center}                                                  
\vspace{5mm}                
\pagestyle{plain}
%
%
\begin{abstract}
Differential privacy provides a theoretical framework for processing a dataset about $n$ users, in a way that the output reveals a minimal information about any single user.
Such notion of privacy is usually ensured by noise-adding mechanisms and amplified by several processes, including subsampling, shuffling, iteration, mixing and diffusion. In this work, we provide privacy amplification bounds for quantum and quantum-inspired algorithms. In particular, we show for the first time, that algorithms running on quantum encoding of a classical dataset or the outcomes of quantum-inspired classical sampling, amplify differential privacy. Moreover, we prove that a quantum version of differential privacy is amplified by the composition of quantum channels, provided that they satisfy some mixing conditions.
\end{abstract}

\section{Introduction}\label{sec:introduction}

Differential Privacy (DP) \cite{dwork1, dwork2} is a rigorous mathematical framework for preserving the information of each individual in a dataset while enabling to analyse and process the dataset. Intuitively, a differentially private algorithm can learn a statistical property of a dataset consisting of $n$ users, yet it leaks \emph{almost} nothing about each individual user. Such mechanisms are of great interest and importance when dealing with sensitive data like hospital data, banks, social media, etc. 
Apart from privacy-preserving data analysis, differential privacy has also found several applications in other fields of computer science such as machine learning \cite{JMLR:v12:chaudhuri11a, abadi2016, papernot2017semisupervised, bassily2018}, statistical learning theory \cite{equiv, JMLR:v17:15-313, online, arunachalam2021private}, mechanism design \cite{design}.

Since its introduction, multiple analytical tools for the design of private data analyses have been developed \cite{boosting, boost2, mult, mult2}. Most commonly, these mechanisms exploit techniques like adding noise to the final output or randomizing the input. A loose analysis of complex mechanisms built out of these blocks can be conducted with simple tools, such as basic composition rules and robustness to post-processing. 
However, the inherent trade-off between privacy and utility in practical applications ignited the development of more refined rules leading to tighter privacy bounds. 
A trend in this direction is to show that several sources of randomness amplify the guarantees of standard DP mechanisms. In particular, DP amplification results have been shown for subsampling, iteration, mixing and shuffling \cite{subsampling, iteration, mixing, shuffling}.

Given the major influence of quantum computing and quantum information in the past decades over different areas of computer science, an interesting question is whether quantum and quantum-inspired algorithms can enhance differential privacy. This question becomes specifically more relevant with the availability of Noisy Intermediate Scale Quantum devices (NISQ) today \cite{preskill_quantum_2018}. The noisy nature of these devices (also previously exploited by \cite{liu2021}) on one hand, and the potential capabilities of quantum algorithms, on the other hand, make such quantum or hybrid quantum-classical mechanisms, an interesting subject of study from the point of view of differential privacy. Furthermore, the connection between machine learning and differential privacy suggests that answering this question can lead to intriguing insights into the capabilities of quantum machine learning. 

Differential privacy has been extended to quantum computation in \cite{quantumDP} and \cite{aaronson2019gentle}. One of the main challenges in translating the definition of DP in the quantum setting is to characterise the notion of \emph{neighbouring quantum states}. Recall that, in the classical setting, two neighbouring datasets differ in at most one single entry. 
In the two mentioned works, 
the adopted definitions of neighbouring quantum states are significantly different, and are respectively based on bounded trace distance and single-register measurements. For the purpose of this paper, we follow the definition of \cite{quantumDP}. 
Moreover, quantum private PAC learning has been defined in \cite{arunachalam2020quantum}, and a quantum analog of the equivalence between private classification and online prediction has been shown in \cite{arunachalam2021private}.

{\paragraph{Our contributions.} In this paper, we initiate a systematic study of differential privacy amplification in quantum and quantum-inspired algorithms. 
We provide three types of results. 
\secref{sec:encoding} provides privacy amplification results when the (classical) data is encoded into a quantum state. 
Informally, we first show that quantum encoding of classical datasets leads to approximate classical differential privacy. Moreover, we show that a quantum DP operation performed on the quantum encoding of a classical dataset satisfies also classical DP, under some suitable assumptions.  These two primary results show a general application of quantum information and quantum encoding for differential privacy and can be employed further to design sophisticated differentially private mechanisms both in the classical and quantum setting. Moreover, we prove that the composition of quantum encoding with noisy mechanisms such as the Laplace and Gaussian mechanisms amplifies differential privacy. 

Similarly, \secref{sec:inspired} investigates the case of \emph{quantum-inspired} algorithms, a family of classical algorithms equipped with an $\ell_2$-norm sampling oracle. We show that differential privacy, both in the exact and approximate setting, is amplified via quantum-inspired subsampling, establishing the concrete amplification bounds.

Finally, our last results concern quantum differential privacy. As for classical DP, quantum DP is preserved under post-processing \cite{quantumDP}. In \secref{sec:pp}, we show an amplification result for the post-processed quantum channel $S\circ T$, provided that $T$ satisfies a quantum analog of the Dobrushin or the Doeblin condition. This is yet another general result that relies on the contraction property of quantum channels and can be exploited to enhance differential privacy. We expand this result by finding explicit bounds for differential privacy under the mentioned condition for the composition of two well-known quantum channels, namely the \emph{generalized amplitude damping channel} and \emph{depolarizing channel}.
Furthermore, we show that the Dobrushin condition together with unitality provides pure quantum differential privacy.

\section{Preliminaries} \label{sec:pre}
We start by introducing notation and concepts that will be used throughout the paper.
\paragraph{Quantum information.} We briefly review the basic concepts in quantum information that we will use throughout the paper.
A $d$-dimensional \textit{pure state} is a unit vector in $\mathbb{C}^{d}$,
written in ket notation as%
\[
\left\vert \psi\right\rangle =\sum_{i=1}^{d}\alpha_{i}\left\vert
i\right\rangle .
\]
Here $\left\vert 1\right\rangle ,\ldots,\left\vert d\right\rangle $\ is an
orthonormal basis for $\mathbb{C}^{d}$,\ and the $\alpha_{i}$'s are complex
numbers called \textit{amplitudes} satisfying $\left\vert \alpha
_{1}\right\vert ^{2}+\cdots+\left\vert \alpha_{d}\right\vert ^{2}=1$. \
The notation $\ket{\cdot}$ reminds of the fact that the Hilbert space has an inner product $\braket{\cdot,\cdot}$, which for Hilbert spaces describing quantum systems is denoted as $\braket{\cdot|\cdot}$.
The left side of the inner product $\bra{\psi}$ is the conjugate transpose of the quantum state $\ket{\psi}$. Then the unit-norm condition can be expressed succinctly as $\braket{\psi|\psi}=1$.


In general, we may also have classical probability distributions over pure states. This scenario is captured by the formalism of \emph{mixed} states, which generalises all the states in quantum mechanics including pure states. Mixed states are conveniently described by density matrices. Formally, a $d$-dimensional mixed state $\rho$ is a $d \times d$ positive semidefinite matrix that satisfies $\text{Tr}(\rho)=1$. Equivalently, $\rho$ is a convex combination of outer products of pure states with themselves:
\[ \rho = \sum_{i=1}^d p_i\ket{\psi_i}\bra{\psi_i},\]
where $p_i \geq 0$ and $\sum_i p_i =1$. In the special case where $p_i=1$, we obtain a pure state $\rho = \ket{\psi_i}\bra{\psi_i}$.
Several norms and distance measures can also be defined for general quantum states. Given a Hermitian matrix $A$ with eigenvalues $\lambda_1,\dots,\lambda_k$, its trace norm is defined as $||A||_\tr:=\frac{1}{2}\sum_{i=1}^d\lambda_i$.
The trace norm induces the \emph{trace distance} $||\rho-\sigma||_{\tr}=Tr(|\rho-\sigma|)/2$.
For a pair of pure states, the trace distance can be linked to their inner product, 
\begin{equation}
\label{eq:inner}
    ||\ket{\psi}\bra{\psi}-\ket{\phi}\bra{\phi}||_{\tr}=\sqrt{1-|\braket{\psi|\phi}|^2}.
\end{equation}


A \emph{superoperator} $S$ maps a mixed state $\rho$ to the mixed state $S(\rho) = \sum_{i=1}^k B_i \rho B_i^\dag$,
where $B_1,\ldots,B_k$ can be any matrices satisfying $\sum_{i=1}^k B_i^\dag B_i = \mathbb{I}.$
This is the most general (norm-preserving) mapping from mixed states to mixed states allowed by quantum mechanics. 
If we drop the norm-preserving condition and we let $\sum_{i=1}^k B_i^\dag B_i \preceq \mathbb{I} $, then we call $S$ a \emph{quantum operation}.
Quantum operations act linearly on mixed states, in the sense that for any $a,b \in \mathbb{C}$, $S(a\rho + b\sigma)=aS(\rho)+bS(\sigma)$.
Moreover, quantum operations are \emph{non trace-increasing}. For any hermitian matrix $A$, $||S(A)||_{\tr} \leq ||A||_{\tr}$. In particular, $||S(\rho)-S(\sigma)||_{\tr} \leq ||\rho -\sigma||_{\tr}$.

The most general class of measurements to perform on mixed states are the POVMs (Positive Operator Valued Measures). In the POVM formalism, a measurement $M$ with possible outcomes $1,2,\dots k$ is given by a list of $d\times d$ positive semidefinite matrices $E_1,\ldots, E_{k}$, which
satisfy $\sum_i E_i = \mathbb{I}$. The measurement rule is:
\[
\Pr[M(\rho)\text{ returns outcome }i]=Tr(E_i\rho).
\]
Notably, trace distance has also the following physical interpretation
\begin{equation}
\label{eq:trace}
    ||\rho-\sigma||_\tr = \max_M \Pr[M(\rho)\text{ accepts}] - \Pr[M(\sigma)\text{ accepts}], 
\end{equation}
where the maximum is taken over all possible two-outcome measurements.
We also define the infinity norm of a matrix $A$ as the maximum of the absolute row sum value as follows,
\begin{equation}
\label{eq:infnormmatrix}
    ||A||_{\infty} = \max_i \sum^n_j |a_{ij}|, 
\end{equation}

Furthermore, we need to define the operator infinity norm or $\ell^{\infty}$-norm. The vector space $\ell^\infty$ is a sequence space whose elements are the bounded sequences. The $\ell^\infty$ space in a Banach space with respect to the following norm,
\begin{equation}
\label{eq:infnorm-op}
    ||x||_{\infty} = \sup_n |x_n|, 
\end{equation}

The operator norm on the Hilbert space is defined over the space of bounded linear operators as,
\begin{equation}
\label{eq:infnorm-hilbert}
    ||O||_{\infty} = \sup ||Ox|| : \forall ||x|| \leq 1, 
\end{equation}

We also note that for the operator norms, $||.||_1$ is the dual norm of $||.||_{\infty}$ \cite{dobrushin3}.

\paragraph{Differential privacy.}
In the standard model of differential privacy, a trusted curator collects the raw data of the individuals and is responsible for their privacy. On the contrary, in the \emph{local} model the curator is possibly malicious, and hence each individual submits their own privatized data.
More formally, consider a statistical dataset, i.e.  a vector $x = (x_1, \dots, x_n)$ over a domain $X$, where each entry $x_i \in X$ represents information contributed by a single individual. datasets $x$ and $x'$ are neighbors if $x_i \neq x_i'$ for exactly one $i \in [n]$. We denote the neighbor relation with $x \sim x'$.
A randomized algorithm $\mathcal{A}$ is $(\epsilon,\delta)$-differentially private if for any two neighbor datasets $x$, $x'$ and for every subset $F$ of the possible outcomes of $\mathcal{A}$ we have
\[\Pr[\mathcal{A}(x) \in F] \leq e^\epsilon\Pr[\mathcal{A}(x') \in F] +\delta.\]

We denote as \emph{pure} differential privacy the special case where $\delta = 0$, while in the most general case we have \emph{approximate} differential privacy. 
One popular method to ensure $(\epsilon,0)$-DP is the \emph{Laplace mechanism}. Given a function $f:\mathcal{X}^n\xrightarrow{} \mathbb{R}k$ we define its $\ell_1$-sensitivity as $\Delta=\max_{x\sim x'}|f(x')-f(x)|$. The Laplace mechanism consists in adding a random perturbation $\eta$ to $f(x)$, where $\eta \sim \mathsf{Laplace}(\Delta/\epsilon):=\frac{\epsilon}{2\Delta}\exp\left(-|\eta|\frac{\epsilon}{\Delta}\right)$.

An additional widely-used method is the Gaussian mechanism, that ensures $(\epsilon,\delta)$-DP. Given a function $f$ as defined above, the Gaussian mechanism consists in adding a random perturbation $\eta$ to $f(x)$, where $\eta \sim \mathcal{N}(0,\sigma^2):=\frac{1}{\sqrt{2\pi\sigma^2}}\exp\left({-\frac{\eta^2}{2\sigma^2}}\right)$ and $\sigma^2=2 \ln(1.25/\delta)\Delta^2/\epsilon^2$.

We now turn our attention to the \emph{local model}. Following the notation used in \cite{equiv}, we say that a randomized algorithm over datasets is $(\epsilon,\delta)$-local differentially private if it's an $(\epsilon,\delta)$-differentially private
algorithm that takes in input a dataset of size $n=1$. 

The most common mechanism for local differential privacy is \emph{randomized response} (RR). For a dataset $x=(x_1,\dots,x_n)\in\{0,1\}^n$, each user $x_i$ outputs a random bit $z_i$, such that $z_i=x_i$ with probability $\frac{e^\epsilon+\delta}{1+e^\epsilon}$ and $z_i=1-x_i$ with probability $\frac{1-\delta}{1+e^\epsilon}$.
It is easy to see that any algorithm run on $z=(z_1,\dots,z_n)$ is $(\epsilon,\delta)$-local differentially private.

Interestingly, differential privacy is related to several desired learnability properties, including robustness, stability and generalization.
Concerning robustness to adversarial examples, we recall here the result stated in (Lemma 1, \cite{lecuyer2019certified}).
Let $B_p(r):=\{\alpha \in \mathbb{R}^n:||\alpha||_p\leq r\}$ be the $p$-norm ball of radius $r$. For a given classification model $f$ and a fixed input $x\in \mathbb{R}^n$, an attacker is able to craft a successful adversarial example of size $L$ for a given $p$-norm if they find $\alpha \in B_p(L)$ such as $f(x+\alpha)\neq f(x)$. The attacker thus tries to find a small change to $x$ that will significantly change the predicted label.
Now, suppose that a randomized function $A$, with bounded output $A(x)\in[0,b]$, $b\in\mathbb{R}^+$, satisfies $(\epsilon,\delta)$-DP. Then the expected value
of its output meets the following property:
\[
\forall \alpha \in B_p(1) : \mathbb{E}(A(x))\leq e^\epsilon \mathbb{E}(A(x+\alpha))+b\delta
\]
where the expectation is taken over the randomness in $A$.

Following the approach proposed in \cite{quantumDP}, we say that a quantum operation $\mathcal{E}$ is $(\tau, \epsilon, \delta)$-quantum differentially private (QDP), if for every POVM $M$, for all subset $S$ of the possible outcomes, and for all inputs $\rho,\sigma$ such that $||\rho-\sigma||_\tr \leq \tau$,
\begin{align}
\begin{split}
    \Pr[M(\mathcal{E}(\rho))\text{ output is in }S]\leq e^\epsilon 
    \Pr[M(\mathcal{E}(\sigma))\text{ output is in }S] +\delta.
\end{split}
\end{align}

\begin{thm}[Proposition 1, \cite{quantumDP}]
\label{thm:pp}
Let $\mathcal{E}$ be a quantum operation that is $(\tau,\epsilon,\delta)$-QDP. Let $\mathcal{F}$ be an arbitrary
quantum operation. Then the composition of $\mathcal{E}$ and $\mathcal{F}$
\[
\mathcal{F}\circ\mathcal{E}:\rho \mapsto \mathcal{F}(\mathcal{E}(\rho))
\]
is $(\tau,\epsilon,\delta)$-QDP.
\end{thm}

\section{Amplification by quantum encoding}
\label{sec:encoding}
In quantum computation, a classical dataset $x\in \mathcal{X}$ can be mapped to a quantum state with a \emph{data-encoding feature map}, also referred as quantum encoding, that is a classical-to-quantum transformation \[\phi(x)=\ket{\phi(x)}\bra{\phi(x)}=\rho(x).\]

Given a dataset $x = (x_1,\dots,x_n)$, where each $x_i$ is a binary string, one of the most common information encoding strategy is the \emph{basis encoding}, which is described by a uniform superposition of computational basis states
\[
x \mapsto \frac{1}{\sqrt{n}}\sum_{i=1}^n \ket{x_i}.
\]
For a complex value dataset $x\in \mathbb{C}^n$, it's convenient to adopt the \emph{amplitude encoding},
\[
x \mapsto \sum_{i=1}^n  x_i \ket{i}.
\]
We always assume that the input vector $x$ is normalised as $||x||^2=\sum_i |x_i|^2 = 1$. For convenience, we set the following parameter, that will be employed in the following. 
\begin{equation}
\label{eq:delta}
    \Gamma(x):= \max_j |x_j|^2.
\end{equation}

For a dataset $x\in[0,2\pi]^n$ we can define the \emph{rotation encoding},
\[
x \mapsto \sum_{q_1\dots q_n=0}^1 \prod_{k=1}^n \cos(x_k)^{q_k} \sin(x_k)^{1-q_k}\ket{q_1,\dots, q_k}.
\]

As noted in \cite{schuld2021supervised}, a quantum encoding gives rise to a \emph{quantum kernel}, which is the inner product between two data-encoding feature vectors. For any $x,x'\in \mathcal{X}$, the quantum kernel induced by $\phi$ is
\begin{equation}
\label{eq:kappa}
\kappa_\phi(x,x')=||\rho(x)\rho(x')||_\tr= |\braket{\phi(x)|\phi(x')}|^2.
\end{equation}
Since we are dealing with differential privacy, quantum kernels evaluated on neighbor inputs are of particular interest. To this end, we define the \emph{quantum minimum adjacent kernel} \[\hat\kappa_\phi:=\min_{x\sim x'}\kappa_\phi(x,x').\]
The expressions of $\kappa_\phi$ and $\hat\kappa_\phi$ for the quantum encodings defined above can be found in \tableref{table1}. We refer to \cite{schuld2021supervised} for more details and more examples of quantum kernels.

We observe that the minimum adjacent kernels allows us to connect the quantum and classical definition of differential privacy.

\begin{lemma}[Quantum-to-classical DP]
Let $x\in \mathcal{X}$ and let $\mathcal{A}$ be a quantum algorithm that takes as input only $\rho(x)=\ket{\phi(x)}\bra{\phi(x)}$ and perform a $(\sqrt{1-\hat\kappa_\phi},\epsilon,\delta)$-QDP quantum operation $\mathcal{E}$ on $\rho(x)$. Then $\mathcal{A}\circ \rho$ is $(\epsilon,\delta)$-DP.
\end{lemma}

\begin{proof}
Let $x,x'$ two neighbouring datasets and $\ket{\phi(x)},\ket{\phi(x')}$ their corresponding encodings. By definition, their trace distance is upper bounded by the \emph{minimum adjacent kernel},
\begin{align*}
   ||\rho(x)-\rho(x')||_{\tr} \leq \sqrt{1-|\braket{\phi(x)|\phi(x')}|^2}
:= \sqrt{1-\kappa_\phi(x,x')}\leq \sqrt{1-\hat\kappa_\phi}. 
\end{align*}
By definition of \emph{quantum differential privacy}, for any measurement $M$, and for any subset $S$ of the possible outcomes,
\begin{align*}
    \Pr[M(\mathcal{E}(\rho(x)))\text{ output is in }S]\leq e^\epsilon 
    \Pr[M(\mathcal{E}(\rho(x')))\text{ output is in }S] +\delta.
\end{align*}
Since quantum DP is robust to post-processing (\thmref{thm:pp}), the former inequality implies that the algorithm $\mathcal{A}\circ \rho$ is $(\epsilon,\delta)$-DP.
\end{proof}

Moreover, we can show that if $\hat\kappa_\phi$ is larger than $0$, then any quantum algorithm that accesses solely the quantum encoding of a classical dataset satisfies approximate differential privacy.

\begin{table}[t]
\caption{Quantum kernels and quantum minimum adjacent kernels for several quantum encodings. Here $\delta_{x,y}$ is the Kronecker function and $\Gamma(x)$ is the parameter defined in \eqref{eq:delta}.}
\label{table1}
\vskip 0.15in
\begin{center}
\begin{small}
\begin{sc}
\begin{tabular}{lcccr}
\toprule
Encoding $\phi$ & $\kappa_\phi(x,x')$ & $\hat\kappa_\phi$ \\
\midrule
Basis encoding    & $\sum_i\delta_{x_i,x'_i}$ & $1-{1}/{n}$\\
Amplitude encoding & $|x^\dag x'|^2$&  $1 - \Gamma(x)$\\
Rotation encoding & $\prod_i |\cos(x_i-x_i')|^2$& $0$ \\

\bottomrule
\end{tabular}
\end{sc}
\end{small}
\end{center}
\vskip -0.1in
\end{table}

\begin{lemma}[Approximate DP by quantum encoding]
\label{lem:adp}
Let $x\in \mathcal{X}$ and let $\mathcal{A}$ be a quantum algorithm that takes as input only $\rho(x)=\ket{\phi(x)}\bra{\phi(x)}$. Then $\mathcal{A}\circ \rho$ is $(0,\sqrt{1-\hat\kappa_\phi})$-DP.
\end{lemma}
\begin{proof}
Since differential privacy is preserved under post-processing, we can assume that $\mathcal{A}$ consists of a quantum operation $S$ followed by a POVM measurement $M$.
Let $F$ be a subset of the possible outcomes of $M$. We define the two-outcome measurement $M'$ such that $M'$ runs $M$ and accepts if the resulting outcome is in $F$, otherwise it rejects.
Plugging \eqref{eq:kappa} and \eqref{eq:inner} into \eqref{eq:trace}, we get the following bound:
\begin{align*}
    \Pr[M(S\ket{\phi(x)})\in F] - \Pr[M(S\ket{\phi(x')})\in F] \\ = \Pr[M'(S\ket{\phi(x)})\text{ accepts}] - \Pr[M'(S\ket{\phi(x')})\text{ accepts}] \\\leq ||S(\rho(x))-S(\rho(x'))||_{\tr}
    \leq ||\rho(x)-\rho(x')||_{\tr} = \sqrt{1-|\braket{\phi(x)|\phi(x')}|^2} \leq \sqrt{1-\hat\kappa_\phi}.
\end{align*}
Thus the algorithm $\mathcal{A}\circ \rho$ is $(0,\sqrt{1-\hat\kappa_\phi})$-DP.
\end{proof}

\begin{algorithm}[tb]
   \caption{Composition of quantum encoding and a noise-adding mechanism}
   \label{alg:laplace}
\begin{algorithmic}
   \STATE {\bfseries Input:} a dataset $x=(x_1,\dots,x_n)$, a POVM $M$ with outcomes in $\{0,1\}$, a distribution $\mathcal{D}$
    \FOR{$i=1$ {\bfseries to} $m$}
        \STATE {Perform the feature map $x\mapsto \rho(x)=\ket{\phi(x)}\bra{\phi(x)}$}
        \STATE Apply $M$ to $\rho(x)$ and store the output in $y_i$
    \ENDFOR
    \STATE Compute the mean $\mu=\frac{1}{m}\sum_{i=1}^m y_i$ and output $O=\mu + \eta$, where $\eta \sim \mathcal{D}$. 
\end{algorithmic}
\end{algorithm}

So far we have shown that quantum encodings inherently provide approximate differential privacy, under some suitable assumptions.
In the following, we study the interaction of quantum encoding and classical noise-adding mechanisms, as sketched in \algref{alg:laplace}.
In particular, we show DP amplification results for the Laplace and Gaussian mechanisms. 

\begin{thm}[Composition of quantum encoding and Laplace mechanism]
\label{thm:laplace}
Let $x,m,M,\rho$ be as in \algref{alg:laplace}.
Let $\mathcal{D}=\mathsf{Laplace}(\frac{1}{\epsilon}(\sqrt{1-\hat{\kappa}_\phi}+t))$ for any $t\geq0$. Then \algref{alg:laplace} is $(\epsilon,0)$-DP with exponentially high probability in $t$ and $m$.
\end{thm}
\begin{proof}

Let $x,x'$ be two neighbouring datasets. Denote as $y_i'$ and $\mu'$ the outcomes of $M$ on input $\rho(x')$ and their mean, respectively. First, we make the following observation 
\begin{align*}
|\mathbb{E} [M(\rho(x)) - M(\rho(x'))]| = |\Pr[M(\rho(x))=1]-\Pr[M(\rho(x'))=1]| \leq \sqrt{1-\hat{\kappa}_\phi}.
\end{align*}
We apply now the Chernoff-Hoeffding's bound
\begin{align*}
    \Pr\left[\left|\frac{1}{m}\sum_{i=1}^m y_i - \mathbb{E}[M(\rho(x))]\right| \geq \frac{t}{2}\right]\leq 2e^{-mt^2}.
\end{align*}
Thus with probability $1-4e^{-mt^2}$ the means $\mu$ and $\mu'$ are within an additive factor $t + \sqrt{1-\hat{\kappa}_\phi}$.

Since $\eta \sim \mathsf{Laplace}((t + \sqrt{1-\hat{\kappa}_\phi})/\epsilon)$, we have that for any $z\in\mathbb{R}$, 
\begin{align*}
    \Pr[\mu+\eta=z] \leq e^\epsilon \Pr[\mu'+\eta=z]
\end{align*} 
with probability at least $1-4e^{-mt^2}$ . In other words, a Laplace perturbation of parameter $(t + \sqrt{1-\hat{\kappa}_\phi})/\epsilon$ ensures $(\epsilon,0)$-DP
with high probability in $t$ and $m$.
\end{proof}

\begin{thm}[Composition of quantum encoding and Gaussian mechanism]
Let $x,m,M,\rho$ be as in \algref{alg:laplace}.
Let $\mathcal{D}=\mathcal{N}(0,\sigma^2)$ where $\sigma^2=2 \ln(1.25/\delta)(\sqrt{1-\hat\kappa_\phi}+t)^2/\epsilon^2$ for any $t\geq0$. Then \algref{alg:laplace} is $(\epsilon,\delta)$-DP with exponentially high probability in $t$ and $m$.
\end{thm}
\begin{proof}
Let $x,x'$ be two neighbouring datasets. Denote as $y_i'$ and $\mu'$ the outcomes of $M$ on input $\rho(x')$ and their mean, respectively.
By applying the same arguments of the proof of \thmref{thm:laplace}, we show that, with probability $1-4e^{-mt^2}$ the means $\mu$ and $\mu'$ are within an additive factor $t + \sqrt{1-\hat{\kappa}_\phi}$.

Since $\eta \sim \mathcal{N}(0,\sigma^2)$ with $\sigma^2=2 \ln(1.25/\delta)(\sqrt{1-\hat\kappa_\phi}+t)^2/\epsilon^2$, we have that for any $z\in\mathbb{R}$, 
\begin{align*}
    \Pr[\mu+\eta=z] \leq e^\epsilon \Pr[\mu'+\eta=z] +\delta.
\end{align*} 
with probability at least $1-4e^{-mt^2}$ . So we proved that a Gaussian perturbation of parameter $\sigma^2$ ensures $(\epsilon,\delta)$-DP
with high probability in $t$ and $m$.
\end{proof}

We can obtain explicit privacy amplification bounds by combining the values of $\hat\kappa_\phi$ in \tableref{table1} with the results of this section. Remark that, unlike the basis and the amplitude encoding, the rotation encoding provides no privacy amplification, as $\hat\kappa_\phi=0$ in that case.

\section{Amplification by quantum-inspired sampling}
\label{sec:inspired}

In quantum-inspired algorithms \cite{tang1, tang2, tang3, tang4}, we simulate the measurement of $\ket{x}^{\otimes m}$ in the computational basis and process the outcomes with a classical algorithm. 
Quantum-inspired subsampling generalizes the uniform subsampling. Indeed, uniform subsampling can be recovered as a special case when $\Gamma(x)=1/n$.

We will show the intuitive fact that quantum-inspired subsampling amplifies DP. The proof closely follows the one of (Problem 1.b, \cite{ullman}) for uniform subsampling, but we include it here for completeness.
Given a normalised vector $x=(x_1,\dots,x_n)\in \mathbb{C}^n$, let $\ket{x}:=\sum_{i=1}^n x_i\ket{i}$ be the amplitude encoding defined in the previous section.
\begin{thm}[DP amplification by quantum-inspired sampling]
\label{lemma: quantum-inspired}
For any $x\in \mathbb{C}^n$, let $s=(s_1,\dots, s_m)$ be the measurement outcomes in the computational basis of $\ket{x}^{\otimes m}$. Denote $\mathcal{S}$ as the sampling mechanism that maps $x$ into $s$.
Let $\mathcal{A}$ be a $(\epsilon,\delta)$-DP algorithm that takes only $s$ as input.  Then $\mathcal{A'}=\mathcal{A}\circ \mathcal{S}$ is $(\epsilon',\delta')$-DP, with $\epsilon' = \log(1+ (e^\epsilon-1)\Gamma(x) m)$ and $\delta'= \delta \Gamma(x) m$.
\end{thm}
\begin{proof}
We will use $T \subseteq \{1,\dots,n\}$ to denote the identities of the $m$-subsampled elements $s_1,\dots,s_m$ (i.e. their index, not their actual value). Note that $T$ is a random variable, and that the randomness of $\mathcal{A'}:=\mathcal{A}\circ \mathcal{S}$ includes both the randomness of the sample $T$ and the random coins of $\mathcal{A}$. Let $x\sim x'$ be adjacent datasets and assume that $x$ and $x'$ differ only on some row $t$. Let $s$ (or $s'$) be a subsample from $x$ (or $x'$) containing the rows in $T$ . Let $F$ be an arbitrary subset of the range of $\mathcal{A}$). For convenience, define $p= \Gamma(x) m$.
Note that, by definition of quantum amplitude encoding and by union bound,
\begin{align*}
  \Pr[i \in T] \leq m \Pr[\ket{x}\text{ collapses to state }\ket{i}] = |x_i|^2 \leq m \Gamma(x):=p
\end{align*}
To show $(\log(1+ p(e^\epsilon-1)), p\delta)$-DP, we have to bound the ratio
\begin{align*}
    \frac{\Pr[\mathcal{A}'(x)\in F]-p\delta}{\Pr[\mathcal{A}'(x')\in F]}\leq
     \frac{p\Pr[\mathcal{A}(s)\in F|i \in T]+ (1-p)\Pr[\mathcal{A}(s)\in F|i \not\in T]-p\delta}{p\Pr[\mathcal{A}(s')\in F|i \in T]+ (1-p)\Pr[\mathcal{A}(s')\in F|i \not\in T]}
\end{align*}
by $p(1+ (e^\epsilon-1))$. For simplicity, define the quantities
\begin{align*}
    C = \Pr[\mathcal{A}(s)\in F|i \in T]
    \\ C' = \Pr[\mathcal{A}(s')\in F|i \in T]
    \\ D = \Pr[\mathcal{A}(s)\in F|i \not\in T] = \Pr[\mathcal{A}(s')\in F|i \not\in T].
\end{align*}
We can rewrite the ratio as 
\begin{align*}
    \frac{\Pr[\mathcal{A}'(x)\in F]-p\delta}{\Pr[\mathcal{A}'(x')\in F]} = \frac{pC + (1-p)D - p\delta}{p C' + (1-p)D}.
\end{align*}
Now we use the fact that, by $(\epsilon, \delta)$-DP, $C \leq \min\{C',D\} + \delta$. Plugging all together, we get 
\begin{align*}
    pC + (1 - p)D - p\delta 
    \leq p(e^\epsilon \min\{C',D\}) + (1 - p)D\\
    \leq p(\min\{C',D\} + (e^\epsilon-1) \min\{C',D\}) 
    + (1 - p)D \\
    \leq p(C' + (e^\epsilon-1) (pC' + (1 -p)D)) 
    + (1 - p)D \\
    \leq (pC' + (1 - p)D) + p(e^\epsilon-1))(pC' + (1 - p)D)
    \leq (1+ p(e^\epsilon-1))(pC' + (1 - p)D),
\end{align*}
where the third-to-last line follow from $\min\{x,y\} \leq \alpha x + (1 - \alpha)y$ for every $0 \leq \alpha \leq 1$. 
To conclude the proof, we rewrite the ratio and get the desired bound.
\[
\frac{\Pr[\mathcal{A}'(x)\in F]-p\delta}{\Pr[\mathcal{A}'(x')\in F]}\leq 1 + p(e^\epsilon -1).
\]
\end{proof}

If we don't require $\mathcal{A}$ to be $(\epsilon,\delta)$-DP, we obtain the following corollary as a special case of~\thmref{lemma: quantum-inspired}.

\begin{cor}[Approximate DP by quantum-inspired sampling]
For any $x\in \mathbb{C}^n$, let $s = (s_1, \dots, s_m)$ be the measurement outcomes in the computational basis of $\ket{x}^{\otimes m}$.
Denote $\mathcal{S}$ as the sampling mechanism that maps $x$ into $s$. Let $\mathcal{A}$ be an  algorithm 
that takes only $s$ as input. Then $\mathcal{A}\circ \mathcal{S}$ is
$(0, \Gamma(x) m)$-DP.
\end{cor}
\section{Amplification by quantum evolution}
\label{sec:pp}

In this section, we look at quantum operations and how they can amplify differential privacy. First, we show a general result regarding the QDP amplification for distance-decreasing quantum operations and then we explore some explicit examples for certain classes of quantum channels. To establish our results, we first need to characterize quantum channels, in terms of the quantum analogs of the classical mixing conditions introduced in \cite{doeblin, dobrushin}.

\begin{definition}
Let $T:\mathcal{H}\xrightarrow{}\mathcal{H'}$ be a quantum operation and $\gamma \in [0,1]$. We say that $T$ is:
\begin{enumerate}
    \item $\gamma$-Dobrushin if 
    \[
     \sup_{\rho \neq \sigma}\frac{||T(\rho)-T(\sigma)||_{\tr}}{||\rho -\sigma||_{\tr}}\leq \gamma.
    \]
    \item $\gamma$-Doeblin if there exists a quantum operation $T':\mathcal{H}\xrightarrow{}\mathcal{H'}$ such that $T'(X)=Tr[X]Y$ for some $Y\in \mathcal{H'}$ and $T-\gamma T'$ is positive.
\end{enumerate}
\end{definition}

We remark that the $\eta^{Tr}(T):=\inf\{\gamma: T \text{ is }\gamma\text{-Dobrushin}\}$, where $\eta^{Tr}(T)$ is the quantum Dobrushin coefficient introduced in \cite{dobrushin2, dobrushin3}.

\begin{lemma}[adapted from \cite{wolf}, Theorem 8.17]
Let T be a $\gamma$-Doeblin quantum operation with $\gamma \in [0,1]$. Then T is a $(1-\gamma)$-Dobrushin quantum operation.
\end{lemma}

Thus, we show that a post-processed quantum channel $\mathcal{S}\circ \mathcal{E}$ amplifies quantum differential privacy, provided that $T$ is $\gamma$-Dobrushin.

\begin{thm}
\label{thm:qpp}
Let $\mathcal{E}$ be a $\gamma$-Dobrushin quantum operation and let $\mathcal{S}$ be a $(\tau,\epsilon(\tau),0)$-QDP quantum operation, where $\epsilon: \mathbb{R}^{+}\rightarrow [0,1]$. Then $\mathcal{S}\circ \mathcal{E}$ is $(\tau,\epsilon(\gamma\tau),0)$-QDP.
\end{thm}

\begin{proof}
Let $\sigma$ and $\rho$ two states such that $||\rho-\sigma||_{\tr}\leq \tau$.
By definition of $\gamma$-Dobrushin, the trace distance between $\mathcal{E}(\rho)$ and $\mathcal{E}(\sigma)$ can be upper bounded as follows.
\[
||\mathcal{E}(\rho)-\mathcal{E}(\sigma)||_{\tr}\leq \gamma ||\rho-\sigma||_{\tr}\leq \gamma\tau.
\]
The channel $\mathcal{S}$ is $(\gamma\tau, \epsilon(\gamma\tau),0)$-QDP, thus, for any measurement $M$,
\begin{align*}
    \Pr[M(\mathcal{S}(\mathcal{E}(\rho))\in F] \leq \exp(\epsilon(\gamma\tau)) \Pr[M(\mathcal{S}(\mathcal{E}(\sigma))\in F].
\end{align*}
The inequality above shows that $\mathcal{S}\circ\mathcal{E}$ is $(\tau, \epsilon(\gamma\tau),0)$-QDP.
\end{proof}

To better demonstrate the application of this result, we give some explicit examples for different quantum channels.
First, we recall the definitions of some relevant quantum operations \cite{nielsen_chuang_2010}. The \emph{generalized amplitude damping} channel for a single qubit is defined as
\[
\mathcal{E}_{GAD}(\rho)=\sum_{k=0}^3 E_k\rho E_k^\dag
\]
in the $2$-dimensional Hilbert space $\mathcal{H}_2$, where
\begin{align*}
    E_0=\sqrt{p} \begin{bmatrix}
1 & 0 \\
0 & \sqrt{1-\gamma} 
\end{bmatrix},\; 
E_1=\sqrt{p} \begin{bmatrix}
0 & \sqrt{\gamma} \\
0 & 0 
\end{bmatrix},\\
E_2=\sqrt{1-p} \begin{bmatrix}
\sqrt{1-\gamma} & 0 \\
0 & 1 
\end{bmatrix},\; E_3=\sqrt{1-p} \begin{bmatrix}
0 & 0 \\
\sqrt{\gamma} & 0 
\end{bmatrix}.
\end{align*}
and $p$ and $\gamma$ are two parameters. 
The \emph{phase damping} channel for a single qubit is defined by the operator-sum representation
\[
\mathcal{E}_{PD}(\rho)= E_0\rho E_0^\dag + E_1\rho E_1^\dag
\]

in the $2$-dimensional Hilbert space $\mathcal{H}_2$, where
\begin{align*}
E_0=\begin{bmatrix}
\sqrt{1} & 0 \\
0 & \sqrt{\lambda} 
\end{bmatrix},\; E_1= \begin{bmatrix}
0 & 0 \\
0 & \sqrt{\lambda} 
\end{bmatrix}.
\end{align*}
Thus we can compose the channels above to define the \emph{phase-amplitude damping} channel
\[
\mathcal{E}_{PAD}(\rho)=\mathcal{E}_{GAD}(\mathcal{E}_{PD}(\rho)).
\]
As shown in \cite{quantumDP}, the phase-amplitude damping channel is $(d,\epsilon,0)$-QDP, with 
\begin{equation}
\label{eq:PAD}
\epsilon:=\ln\left(1+\frac{2d\sqrt{1-\gamma}\sqrt{1-\lambda}}{1-\sqrt{1-\gamma}\sqrt{1-\lambda}}\right)
\end{equation}
The \emph{depolarizing} channel corresponds to the quantum operation
\[
\mathcal{E}_{Dep}= p\frac{\mathbb{I}}{D} + (1-p)\rho
\]
where $D$ is the dimension of the state Hilbert space and $p$
is the probability parameter.
As shown in \cite{quantumDP}, the depolarizing channel is $(d,\epsilon,0)$-QDP, with 
\[
\epsilon:= \ln\left(1+\frac{1-p}{p}dD\right).
\]
The following result characterizes a family of Dobrushin quantum operation.
\begin{lemma}[\cite{dobrushin3}, Theorem 6.1]
\label{lem:unital}
Let $\Phi_T$ be a quantum operation such that
\begin{enumerate}
    \item $\Phi_T(\mathbb{I})=\mathbb{I}$ ($\Phi_T$ is unital),
    \item $\Phi_T : \mathbb{I} + w\cdot \sigma \rightarrow{} \mathbb{I} + (T w)\cdot \sigma$ where $T$ is a real matrix
with $||T||_\infty \leq 1$, where $||\cdot||_{\infty}$ is the operator norm.
\end{enumerate} 
Then $\Phi_T$ is $||T||_\infty$-Dobrushin.
\end{lemma}
Interestingly, the Dobrushin condition and unitality ensure quantum \emph{pure} differential privacy, as we show in the following theorem.
\begin{thm}
Let $\Phi$ be a quantum operation in the $2$-dimensional Hilbert space $\mathcal{H}_2$, such that
\begin{enumerate}
    \item $\Phi(\mathbb{I})=\mathbb{I}$ ($\Phi$ is unital),
    \item $\Phi$ is $\gamma$-Dobrushin.
\end{enumerate} 
Then $\Phi$ is $(d,\log(1+2d\gamma),0)$-QDP.
\end{thm}
\begin{proof}
Consider two arbitrary qubit states $\rho_1,\rho_2$ and set $||\rho_1-\rho_2||_\tr:=d$. The $\gamma$-Dobrushin condition can be restated as follows:
\[
||\Phi(\rho_1)-\Phi(\rho_2)||_\tr \leq \gamma ||\rho_1-\rho_2||_\tr = \gamma d.
\]
Given an arbitrary POVM $M=\{M_m\}$, we want to bound the following quantity
\[
\frac{\Tr\{\Phi(\rho_1) M_m\}}{\Tr\{\Phi(\rho_2) M_m\}} -1
\]
First, we upper bound the numerator:
\[
{\Tr\{\Phi(\rho_1) M_m\}}-{\Tr\{\Phi(\rho_2) M_m\}} = {\Tr\{(\Phi(\rho_1)-\Phi(\rho_2)) M_m\}} \leq d\gamma \Tr\{M_m\}.
\]
Since $\rho_2$ is a qubit state, can write $\rho_2= \frac{1}{2}(\mathbb{I} + r \cdot \sigma )$ for a Bloch vector $r$.
\[
{\Tr\{\Phi(\rho_2) M_m\}} = \Tr\left\{\frac{1}{2}\Phi(\mathbb{I} + r \cdot \sigma )M_m\right\} \geq \frac{1}{2}\Tr\{\Phi(\mathbb{I})M_m\}= \frac{1}{2}\Tr\{M_m\},
\]
where we used unitality in the last inequality.
Putting all together, we get
\[
\frac{\Tr\{\Phi(\rho_1) M_m\}}{\Tr\{\Phi(\rho_2) M_m\}} -1 \leq  \frac{d\gamma \Tr\{M_m\}}{\frac{1}{2}\Tr\{M_m\}} = 2d\gamma.
\]
Thus the channel $\Phi$ is $\log(1+2d\gamma)$-QDP.
\[
\frac{\Tr\{\Phi(\rho_1) M_m\}}{\Tr\{\Phi(\rho_2) M_m\}} \leq  e^\epsilon,
\]
where $\epsilon:= \ln(1 + 2d\gamma)$.
\end{proof}

\begin{table}[t]
\caption{
The channels below satisfies $(d,\epsilon,0)$-QDP with the $\epsilon$ values shown in the table. We denoted as $\Phi_{\gamma}$ an arbitrary $\gamma$-Dobrushin unital channel.}
\label{sample-table}
\vskip 0.15in
\begin{center}
\begin{small}
\begin{sc}
\begin{tabular}{lcccr}
\toprule
Channel & $\epsilon$ & Reference \\
\midrule
$\mathcal{E}_{Dep}$    & $\ln\left(1+\frac{1-p}{p}dD\right)$ &\cite{quantumDP}\\
$\mathcal{E}_{PAD}$& $\ln\left(1+\frac{2d\sqrt{1-\gamma}\sqrt{1-\lambda}}{1-\sqrt{1-\gamma}\sqrt{1-\lambda}}\right)$ &\cite{quantumDP}\\
$\Phi_{\gamma}$&$\ln(1 + 2d\gamma)$& Theorem 6\\
$\mathcal{E}_{PAD}\circ \mathcal{E}_{Dep}$ & $(1-p)\ln\left(1+\frac{2d\sqrt{1-\gamma}\sqrt{1-\lambda}}{1-\sqrt{1-\gamma}\sqrt{1-\lambda}}\right)$& Theorem 7\\
\bottomrule
\end{tabular}
\end{sc}
\end{small}
\end{center}
\vskip -0.1in
\end{table}
Finally, we show how \thmref{thm:qpp} can be used to derive privacy amplification bounds for the composition of several channels.
\begin{thm}
The composition of the depolarizing channel and the phase-amplitude damping channel $\mathcal{E}_{PAD}\circ \mathcal{E}_{Dep}$ is $(d,\epsilon,0)$-QDP, where
\[
\epsilon = (1-p)\ln\left(1+\frac{2d\sqrt{1-\gamma}\sqrt{1-\lambda}}{1-\sqrt{1-\gamma}\sqrt{1-\lambda}}\right).
\]

\end{thm}
\begin{proof}
We use the fact that the depolarizing channel $\mathcal{E}_{Dep}$ satisfies the Dobrushin condition. This can be shown either by direct computation or by observing that  $\mathcal{E}_{Dep}$ satisfies the hypothesis of \lemref{lem:unital} with $||T||_\infty = 1-p$.
Then we can combine \eqref{eq:PAD} with \thmref{thm:qpp} to derive the desired bound for the composed channel $\mathcal{E}_{PAD}\circ \mathcal{E}_{Dep}$.
\end{proof}
\section{Discussion and future work}\label{sec:conclusion}

We have undertaken a systematic study of differential privacy amplification in quantum and quantum-inspired algorithms. Our work is the first to reason about quantum encodings through the lens of differential privacy, laying the foundation for further analysis. Prior to this work, the choice of the encoding was motivated mainly by expressiveness, efficiency and robustness to experimental noise \cite{coyle, schuld2021}. Due to the intimate relation between DP, algorithmic stability and robustness to adversarial examples, our results suggest new criteria for the choice of quantum encodings.
Previous work explored the relation between experimental noise, quantum differential privacy and robustness. The tighter bounds presented in our paper can be used to improve the result of \cite{liu2021}, taking into account the composition of several mechanisms.

In the future, it would interesting to provide similar amplification results for the notion of quantum differential privacy employed in \cite{aaronson2019gentle}. As previously mentioned, in the quantum setting, different definitions of neighbouring quantum states lead to different notions for quantum differential privacy. Thus, understanding the relationship between these notions is of both theoretical and practical interest. To this end, one could also adopt the variety of quantum distances available in the quantum information literature. One of the best candidates for this purpose is the \emph{quantum Wasserstein distance} introduced in \cite{wasserstein2021}, which generalises the notion of neighbouring quantum states of \cite{aaronson2019gentle}. Furthermore, the relation between the contraction coefficient of quantum channels and DP which we have explored in this paper can also be studied alternatively with this distance due to the contractivity results proved in \cite{wasserstein2021}.

Another interesting future direction would be to expand our composition results and the previous work of \cite{quantumDP} on the differential privacy of specific quantum channels, to more general classes of quantum operations, in term of their general characteristics such as contraction coefficients or channel capacity. On this note, a good candidate would be to study the effect of the class of LOCC (Local Operations and Classical Communication) operations on both classical and quantum differential privacy. This class is of particular interest due to its relation to entanglement, which is another non-classical and unique property of the quantum world to be studied in the context of differential privacy.

Finally, as a follow up of our theoretical results, we aim to investigate the experimental implementation, their feasibility, and their application using the available NISQ devices.
Experimental noise is among the major limitations of current architectures. Yet, demonstrating that such noise provides beneficial properties, such as privacy and robustness, could shape the pathway for new applications.

\paragraph{Acknowledgements.}
{We thank Vincent Cohen-Addad, Alex B. Grilo, Nai-Hui Chia and Brian Coyle for useful discussions.}

\bibliography{qdp}

\begin{thebibliography}{40}
\providecommand{\natexlab}[1]{#1}
\providecommand{\url}[1]{\texttt{#1}}
\expandafter\ifx\csname urlstyle\endcsname\relax
  \providecommand{\doi}[1]{doi: #1}\else
  \providecommand{\doi}{doi: \begingroup \urlstyle{rm}\Url}\fi

\bibitem[Dwork et~al.(2006)Dwork, McSherry, Nissim, and Smith]{dwork1}
Cynthia Dwork, Frank McSherry, Kobbi Nissim, and Adam Smith.
\newblock Calibrating noise to sensitivity in private data analysis.
\newblock In \emph{Proceedings of the Third Conference on Theory of
  Cryptography}, TCC'06, page 265–284, Berlin, Heidelberg, 2006.
  Springer-Verlag.
\newblock ISBN 3540327312.
\newblock \doi{10.1007/11681878_14}.
\newblock URL \url{https://doi.org/10.1007/11681878_14}.

\bibitem[Dwork and Roth(2014)]{dwork2}
Cynthia Dwork and Aaron Roth.
\newblock The algorithmic foundations of differential privacy.
\newblock 9\penalty0 (3–4):\penalty0 211–407, August 2014.
\newblock ISSN 1551-305X.
\newblock \doi{10.1561/0400000042}.
\newblock URL \url{https://doi.org/10.1561/0400000042}.

\bibitem[Chaudhuri et~al.(2011)Chaudhuri, Monteleoni, and
  Sarwate]{JMLR:v12:chaudhuri11a}
Kamalika Chaudhuri, Claire Monteleoni, and Anand~D. Sarwate.
\newblock Differentially private empirical risk minimization.
\newblock \emph{Journal of Machine Learning Research}, 12\penalty0
  (29):\penalty0 1069--1109, 2011.
\newblock URL \url{http://jmlr.org/papers/v12/chaudhuri11a.html}.

\bibitem[Abadi et~al.(2016)Abadi, Chu, Goodfellow, McMahan, Mironov, Talwar,
  and Zhang]{abadi2016}
Martin Abadi, Andy Chu, Ian Goodfellow, H.~Brendan McMahan, Ilya Mironov, Kunal
  Talwar, and Li~Zhang.
\newblock Deep learning with differential privacy.
\newblock \emph{Proceedings of the 2016 ACM SIGSAC Conference on Computer and
  Communications Security}, Oct 2016.
\newblock \doi{10.1145/2976749.2978318}.
\newblock URL \url{http://dx.doi.org/10.1145/2976749.2978318}.

\bibitem[Papernot et~al.(2017)Papernot, Abadi, Úlfar Erlingsson, Goodfellow,
  and Talwar]{papernot2017semisupervised}
Nicolas Papernot, Martín Abadi, Úlfar Erlingsson, Ian Goodfellow, and Kunal
  Talwar.
\newblock Semi-supervised knowledge transfer for deep learning from private
  training data, 2017.

\bibitem[Bassily et~al.(2018)Bassily, Thakkar, and Thakurta]{bassily2018}
Raef Bassily, Om~Thakkar, and Abhradeep Thakurta.
\newblock Model-agnostic private learning.
\newblock In \emph{Proceedings of the 32nd International Conference on Neural
  Information Processing Systems}, NIPS'18, page 7102–7112, Red Hook, NY,
  USA, 2018. Curran Associates Inc.

\bibitem[Kasiviswanathan et~al.(2011)Kasiviswanathan, Lee, Nissim,
  Raskhodnikova, and Smith]{equiv}
Shiva~Prasad Kasiviswanathan, Homin~K. Lee, Kobbi Nissim, Sofya Raskhodnikova,
  and Adam Smith.
\newblock What can we learn privately?
\newblock \emph{SIAM J. Comput.}, 40\penalty0 (3):\penalty0 793–826, June
  2011.
\newblock ISSN 0097-5397.
\newblock \doi{10.1137/090756090}.
\newblock URL \url{https://doi.org/10.1137/090756090}.

\bibitem[Wang et~al.(2016)Wang, Lei, and Fienberg]{JMLR:v17:15-313}
Yu-Xiang Wang, Jing Lei, and Stephen~E. Fienberg.
\newblock Learning with differential privacy: Stability, learnability and the
  sufficiency and necessity of erm principle.
\newblock \emph{Journal of Machine Learning Research}, 17\penalty0
  (183):\penalty0 1--40, 2016.
\newblock URL \url{http://jmlr.org/papers/v17/15-313.html}.

\bibitem[Bun et~al.(2020)Bun, Livni, and Moran]{online}
M.~Bun, R.~Livni, and S.~Moran.
\newblock An equivalence between private classification and online prediction.
\newblock In \emph{2020 IEEE 61st Annual Symposium on Foundations of Computer
  Science (FOCS)}, pages 389--402, Los Alamitos, CA, USA, nov 2020. IEEE
  Computer Society.
\newblock \doi{10.1109/FOCS46700.2020.00044}.
\newblock URL
  \url{https://doi.ieeecomputersociety.org/10.1109/FOCS46700.2020.00044}.

\bibitem[Arunachalam et~al.(2021)Arunachalam, Quek, and
  Smolin]{arunachalam2021private}
Srinivasan Arunachalam, Yihui Quek, and John Smolin.
\newblock Private learning implies quantum stability.
\newblock In \emph{Advances in Neural Information Processing Systems 34
  pre-proceedings (NeurIPS 2021)}, NIPS'21, 2021.

\bibitem[McSherry and Talwar(2007)]{design}
Frank McSherry and Kunal Talwar.
\newblock Mechanism design via differential privacy.
\newblock In \emph{48th Annual IEEE Symposium on Foundations of Computer
  Science (FOCS'07)}, pages 94--103, 2007.
\newblock \doi{10.1109/FOCS.2007.66}.

\bibitem[Dwork et~al.(2010)Dwork, Rothblum, and Vadhan]{boosting}
Cynthia Dwork, Guy~N. Rothblum, and Salil Vadhan.
\newblock Boosting and differential privacy.
\newblock In \emph{2010 IEEE 51st Annual Symposium on Foundations of Computer
  Science}, pages 51--60, 2010.
\newblock \doi{10.1109/FOCS.2010.12}.

\bibitem[Kairouz et~al.(2017)Kairouz, Oh, and Viswanath]{boost2}
Peter Kairouz, Sewoong Oh, and Pramod Viswanath.
\newblock The composition theorem for differential privacy.
\newblock \emph{IEEE Transactions on Information Theory}, 63\penalty0
  (6):\penalty0 4037--4049, 2017.
\newblock \doi{10.1109/TIT.2017.2685505}.

\bibitem[Hardt and Rothblum(2010)]{mult}
Moritz Hardt and Guy~N. Rothblum.
\newblock A multiplicative weights mechanism for privacy-preserving data
  analysis.
\newblock In \emph{2010 IEEE 51st Annual Symposium on Foundations of Computer
  Science}, pages 61--70, 2010.
\newblock \doi{10.1109/FOCS.2010.85}.

\bibitem[Hardt et~al.(2012)Hardt, Ligett, and McSherry]{mult2}
Moritz Hardt, Katrina Ligett, and Frank McSherry.
\newblock A simple and practical algorithm for differentially private data
  release.
\newblock In \emph{Proceedings of the 25th International Conference on Neural
  Information Processing Systems - Volume 2}, NIPS'12, page 2339–2347, Red
  Hook, NY, USA, 2012. Curran Associates Inc.

\bibitem[Balle et~al.(2018)Balle, Barthe, and Gaboardi]{subsampling}
Borja Balle, Gilles Barthe, and Marco Gaboardi.
\newblock Privacy amplification by subsampling: Tight analyses via couplings
  and divergences.
\newblock In \emph{Proceedings of the 32nd International Conference on Neural
  Information Processing Systems}, NIPS'18, page 6280–6290, Red Hook, NY,
  USA, 2018. Curran Associates Inc.

\bibitem[Feldman et~al.(2018)Feldman, Mironov, Talwar, and Thakurta]{iteration}
Vitaly Feldman, Ilya Mironov, Kunal Talwar, and Abhradeep Thakurta.
\newblock Privacy amplification by iteration.
\newblock \emph{2018 IEEE 59th Annual Symposium on Foundations of Computer
  Science (FOCS)}, Oct 2018.
\newblock \doi{10.1109/focs.2018.00056}.
\newblock URL \url{http://dx.doi.org/10.1109/FOCS.2018.00056}.

\bibitem[Balle et~al.(2019)Balle, Barthe, Gaboardi, and Geumlek]{mixing}
Borja Balle, Gilles Barthe, Marco Gaboardi, and Joseph Geumlek.
\newblock \emph{Privacy Amplification by Mixing and Diffusion Mechanisms}.
\newblock Curran Associates Inc., Red Hook, NY, USA, 2019.

\bibitem[Cheu et~al.(2019)Cheu, Smith, Ullman, Zeber, and Zhilyaev]{shuffling}
Albert Cheu, Adam Smith, Jonathan Ullman, David Zeber, and Maxim Zhilyaev.
\newblock Distributed differential privacy via shuffling.
\newblock \emph{Lecture Notes in Computer Science}, page 375–403, 2019.
\newblock ISSN 1611-3349.
\newblock \doi{10.1007/978-3-030-17653-2_13}.
\newblock URL \url{http://dx.doi.org/10.1007/978-3-030-17653-2_13}.

\bibitem[Preskill(2018)]{preskill_quantum_2018}
John Preskill.
\newblock Quantum {Computing} in the {NISQ} era and beyond.
\newblock \emph{Quantum}, 2:\penalty0 79, August 2018.
\newblock \doi{10.22331/q-2018-08-06-79}.
\newblock URL \url{https://quantum-journal.org/papers/q-2018-08-06-79/}.
\newblock Publisher: Verein zur F{\"o}rderung des Open Access Publizierens in
  den Quantenwissenschaften.

\bibitem[Du et~al.(2021)Du, Hsieh, Liu, Tao, and Liu]{liu2021}
Yuxuan Du, Min-Hsiu Hsieh, Tongliang Liu, Dacheng Tao, and Nana Liu.
\newblock Quantum noise protects quantum classifiers against adversaries.
\newblock \emph{Physical Review Research}, 3\penalty0 (2), May 2021.
\newblock ISSN 2643-1564.
\newblock \doi{10.1103/physrevresearch.3.023153}.
\newblock URL \url{http://dx.doi.org/10.1103/PhysRevResearch.3.023153}.

\bibitem[Zhou and Ying(2017)]{quantumDP}
Li~Zhou and Mingsheng Ying.
\newblock Differential privacy in quantum computation.
\newblock In \emph{2017 IEEE 30th Computer Security Foundations Symposium
  (CSF)}, pages 249--262, 2017.
\newblock \doi{10.1109/CSF.2017.23}.

\bibitem[Aaronson and Rothblum(2019)]{aaronson2019gentle}
Scott Aaronson and Guy~N. Rothblum.
\newblock Gentle measurement of quantum states and differential privacy.
\newblock In \emph{Proceedings of the 51st Annual ACM SIGACT Symposium on
  Theory of Computing}, STOC 2019, page 322–333, New York, NY, USA, 2019.
  Association for Computing Machinery.
\newblock ISBN 9781450367059.
\newblock \doi{10.1145/3313276.3316378}.
\newblock URL \url{https://doi.org/10.1145/3313276.3316378}.

\bibitem[Arunachalam et~al.(2020)Arunachalam, Grilo, and
  Yuen]{arunachalam2020quantum}
Srinivasan Arunachalam, Alex~B. Grilo, and Henry Yuen.
\newblock Quantum statistical query learning, 2020.
\newblock URL \url{https://arxiv.org/abs/2002.08240}.

\bibitem[Hiai and Ruskai(2016)]{dobrushin3}
Fumio Hiai and Mary~Beth Ruskai.
\newblock Contraction coefficients for noisy quantum channels.
\newblock \emph{Journal of Mathematical Physics}, 57\penalty0 (1):\penalty0
  015211, Jan 2016.
\newblock ISSN 1089-7658.
\newblock \doi{10.1063/1.4936215}.
\newblock URL \url{http://dx.doi.org/10.1063/1.4936215}.

\bibitem[Lecuyer et~al.(2019)Lecuyer, Atlidakis, Geambasu, Hsu, and
  Jana]{lecuyer2019certified}
Mathias Lecuyer, Vaggelis Atlidakis, Roxana Geambasu, Daniel Hsu, and Suman
  Jana.
\newblock Certified robustness to adversarial examples with differential
  privacy, 2019.

\bibitem[Schuld(2021)]{schuld2021supervised}
Maria Schuld.
\newblock Supervised quantum machine learning models are kernel methods, 2021.
\newblock URL \url{https://arxiv.org/abs/2101.11020}.

\bibitem[Tang(2019)]{tang1}
Ewin Tang.
\newblock A quantum-inspired classical algorithm for recommendation systems.
\newblock In \emph{Proceedings of the 51st Annual ACM SIGACT Symposium on
  Theory of Computing}, STOC 2019, page 217–228, New York, NY, USA, 2019.
  Association for Computing Machinery.
\newblock ISBN 9781450367059.
\newblock \doi{10.1145/3313276.3316310}.
\newblock URL \url{https://doi.org/10.1145/3313276.3316310}.

\bibitem[Tang(2021)]{tang2}
Ewin Tang.
\newblock Quantum principal component analysis only achieves an exponential
  speedup because of its state preparation assumptions.
\newblock \emph{Physical Review Letters}, 127\penalty0 (6), Aug 2021.
\newblock ISSN 1079-7114.
\newblock \doi{10.1103/physrevlett.127.060503}.
\newblock URL \url{http://dx.doi.org/10.1103/PhysRevLett.127.060503}.

\bibitem[Gilyén et~al.(2018)Gilyén, Lloyd, and Tang]{tang3}
András Gilyén, Seth Lloyd, and Ewin Tang.
\newblock Quantum-inspired low-rank stochastic regression with logarithmic
  dependence on the dimension, 2018.

\bibitem[Chia et~al.(2018)Chia, Lin, and Wang]{tang4}
Nai-Hui Chia, Han-Hsuan Lin, and Chunhao Wang.
\newblock Quantum-inspired sublinear classical algorithms for solving low-rank
  linear systems, 2018.
\newblock URL \url{https://arxiv.org/abs/1811.04852}.

\bibitem[Ullman(2017)]{ullman}
Jonathan Ullman.
\newblock Cs7880: Rigorous approaches to data privacy, 2017.
\newblock URL \url{https://www.ccs.neu.edu/home/jullman/cs7880s17/HW1sol.pdf}.

\bibitem[Doeblin(1937)]{doeblin}
W.~Doeblin.
\newblock Sur les proprietes asymptotiques de mouvements régis par certains
  types de chaînes simples (suite et fin).
\newblock \emph{Bulletin mathématique de la Société Roumaine des Sciences},
  39\penalty0 (2):\penalty0 3--61, 1937.
\newblock ISSN 12203858.
\newblock URL \url{http://www.jstor.org/stable/43769812}.

\bibitem[Dobrushin(1956)]{dobrushin}
R.~L. Dobrushin.
\newblock Central limit theorem for nonstationary markov chains. ii.
\newblock \emph{Theory of Probability \& Its Applications}, 1\penalty0
  (4):\penalty0 329--383, 1956.
\newblock \doi{10.1137/1101029}.
\newblock URL \url{https://doi.org/10.1137/1101029}.

\bibitem[Gaubert and Qu(2015)]{dobrushin2}
St{\'e}phane Gaubert and Zheng Qu.
\newblock {Dobrushin ergodicity coefficient for Markov operators on cones}.
\newblock \emph{{Integral Equations and Operator Theory}}, 1\penalty0
  (81):\penalty0 127--150, January 2015.
\newblock \doi{10.1007/s00020-014-2193-2}.
\newblock URL \url{https://hal.inria.fr/hal-01099179}.
\newblock Also arXiv:1307.4649.

\bibitem[Wolf(2012)]{wolf}
Michael~M. Wolf.
\newblock Quantum channels and operations. guided tour, July 2012.
\newblock URL
  \url{https://www-m5.ma.tum.de/foswiki/pub/M5/Allgemeines/MichaelWolf/QChannelLecture.pdf}.

\bibitem[Nielsen and Chuang(2010)]{nielsen_chuang_2010}
Michael~A. Nielsen and Isaac~L. Chuang.
\newblock \emph{Quantum Computation and Quantum Information: 10th Anniversary
  Edition}.
\newblock Cambridge University Press, 2010.
\newblock \doi{10.1017/CBO9780511976667}.

\bibitem[LaRose and Coyle(2020)]{coyle}
Ryan LaRose and Brian Coyle.
\newblock Robust data encodings for quantum classifiers.
\newblock \emph{Phys. Rev. A}, 102:\penalty0 032420, Sep 2020.
\newblock \doi{10.1103/PhysRevA.102.032420}.
\newblock URL \url{https://link.aps.org/doi/10.1103/PhysRevA.102.032420}.

\bibitem[Schuld et~al.(2021)Schuld, Sweke, and Meyer]{schuld2021}
Maria Schuld, Ryan Sweke, and Johannes~Jakob Meyer.
\newblock Effect of data encoding on the expressive power of variational
  quantum-machine-learning models.
\newblock \emph{Physical Review A}, 103\penalty0 (3), Mar 2021.
\newblock ISSN 2469-9934.
\newblock \doi{10.1103/physreva.103.032430}.
\newblock URL \url{http://dx.doi.org/10.1103/PhysRevA.103.032430}.

\bibitem[De~Palma et~al.(2021)De~Palma, Marvian, Trevisan, and
  Lloyd]{wasserstein2021}
Giacomo De~Palma, Milad Marvian, Dario Trevisan, and Seth Lloyd.
\newblock The quantum wasserstein distance of order 1.
\newblock \emph{IEEE Transactions on Information Theory}, 67\penalty0
  (10):\penalty0 6627–6643, Oct 2021.
\newblock ISSN 1557-9654.
\newblock \doi{10.1109/tit.2021.3076442}.
\newblock URL \url{http://dx.doi.org/10.1109/TIT.2021.3076442}.

\end{thebibliography}
\bibliographystyle{unsrtnat}

\end{document}